\documentclass[11pt]{article}
\usepackage{arxivbasic}
\usepackage{graphicx}
\usepackage[title]{appendix}  %% More flexible appendix formatting

\usepackage{amsmath, amsthm, amssymb, mathdots} 
\allowdisplaybreaks % allow align equations to be split across page breaks

\newtheorem{theorem}{Theorem}

\newtheorem{corollary}{Corollary} % above says to follow Theorem counter. This gives Corollary its own counter.
\theoremstyle{definition}
\usepackage{dsfont} % For a double-stroke 1 to write indicator functions (aka the characteristic function for a set)

\usepackage{enumitem}

\usepackage[font=small,labelfont=bf]{caption}

\usepackage{microtype}
 \DisableLigatures{encoding = *, family = * }

\usepackage{hyperref}
\hypersetup{colorlinks=true, citecolor=black, linkcolor=black, urlcolor=black}

\usepackage{listings}
\usepackage[usenames,dvipsnames]{color}    

\usepackage[
	backend=bibtex,
	style=authoryear,
	citestyle=authoryear-comp,
	natbib=true,
	sorting=nyvt,
	doi=true,isbn=false,url=false
  ]{biblatex}
\begin{document}

\title{Building Mean Field State Transition Models Using The Generalized Linear Chain Trick and Continuous Time Markov Chain Theory}

% order TBD -- I went with alphabetical
\author{
    Hurtado, Paul J. \\
	University of Nevada, Reno \\
	ORCID: 0000-0002-8499-5986 \\
	\texttt{phurtado@unr.edu}
	\And
	Richards, Cameron\\
	University of Nevada, Reno\\ 
	ORCID: 0000-0002-1620-9998\\
}

%\date{Last Updated: \today}

\maketitle

\begin{abstract} The well-known Linear Chain Trick (LCT) allows modelers to derive mean field ODEs that assume gamma (Erlang) distributed passage times, by transitioning individuals sequentially through a chain of sub-states. The time spent in these states is the sum of $k$ exponentially distributed random variables, and is thus gamma (Erlang) distributed. The Generalized Linear Chain Trick (GLCT) extends this technique to the much broader phase-type family of distributions, which includes exponential, Erlang, hypoexponential, and Coxian distributions. Intuitively, phase-type distributions are the absorption time distributions for continuous time Markov chains (CTMCs). Here we review CTMCs and phase-type distributions, then illustrate how to use the GLCT to efficiently build mean field ODE models from underlying stochastic model assumptions. We generalize the Rosenzweig-MacArthur and SEIR models and show the benefits of using the GLCT to compute numerical solutions. These results highlight some practical benefits, and the intuitive nature, of using the GLCT to derive ODE models from first principles.

\end{abstract}


\clearpage
\tableofcontents
\clearpage

\section{Introduction} \label{sec:intro} 

Continuous time state transition models, often formulated as mean field ODE models, are widely used throughout the biological sciences and across multiple scales. Examples include models of multi-species interactions, infectious disease transmission, cell proliferation, and various other applications in which entities transition among a finite number of states \citep[e.g.,][]{AndersonMay1992,Clapp2015,Strogatz2014,Yates2017,Meiss2017,ArrowsmithPlace1990,Allen2007,EdelsteinKeshet2005,Dayan2005,Izhikevich2010,HSD2013,EllnerGuckenheimer2006,Murray2007,Murray2011,Wiggins2003,Keener2008a,Keener2008b,BeuterGlassMackeyTitcombe2003}.  One criticism of mean field ODE models is that they often implicitly assume the time individuals spend in the different states are exponentially distributed, and it is known that the timing of state transitions can very meaningfully affect model dynamics and model outputs in an applied setting \citep{Wearing2005, Robertson2018, Metz1986, Metz1991, Nisbet1989, Krylova2013,Getz2018b}. That is, an ODE model with a linear loss rate can be interpreted as implicitly assuming an underlying stochastic state transmission model with an exponentially distributed dwell-time in that focal state. For example, the simple model $dx/dt =\, b - d\,x$ is consistent with assuming an underlying stochastic state transition model in which individuals spend an exponentially distributed amount of time with mean $1/d$ in the state corresponding to variable $x$. 

One remedy to address this issue with ODE models is known as the Linear Chain Trick \citep[LCT; ][and references therein]{Smith2010,Hurtado2019}, which allows modelers to instead assume gamma (Erlang\footnote{Gamma distributions with integer-valued shape parameters are those that can be thought of as the sum of \textit{iid} exponentially distributed random variables, and are known as \textit{Erlang distributions.}}) distributed passage times (a.k.a., \textit{dwell times}). This is accomplished by partitioning a state into a series of $k$ sub-states, where individuals transition sequentially through this ``linear chain" of states. The resulting time spent in this collection of sub-states is thus the sum of $k$ exponentially distributed random variables, and therefore follows an Erlang distribution (if each exponential has the same rate) or a generalized Erlang\footnote{The sum of independent exponentially distributed random variables with different rates is known as a generalized Erlang or hypoexponential distribution.} distribution if the rates differ. 

The Generalized Linear Chain Trick (GLCT) \citep{Hurtado2019} extends this technique to allow modelers to assume these passage times follow a much broader family of distributions that includes the \textit{phase-type} family of distributions \citep{Bladt2017,Bladt2017ch3,Reinecke2012a,Horvath2016,BuTools2}. This broad family includes exponential, Erlang, hypoexponential, hyperexponential, Coxian and some other named distributions. Intuitively, phase-type distributions can be thought of as the family of all possible \textit{hitting time} (or \textit{absorption time}) distributions for continuous time Markov chains (CTMCs). In addition, statistical methods exist for estimating such distributions from data \citep[][and references therein]{Horvath2012,Horvath2016,Hurtado2019} allowing researchers to build approximate empirical distributions into ODE models using a more flexible family of distributions than only the Erlang distributions.

In this paper, we illustrate how to use the GLCT alongside concepts and techniques from CTMC theory to build and numerically solve mean field ODE models using a much richer set of possible underlying stochastic model assumptions. The paper is organized as follows. First, we review CTMCs and phase-type distributions. We then state the GLCT for phase-type distributions and, for comparison, the well-known LCT. In the Results section, we generalize some simple biological state transition models by replacing their implicit assumption of exponentially distributed passage time assumptions with arbitrary phase-type distributions. Lastly, we investigate some of the computational costs and benefits of using this generalized model framework with regards to computing numerical solutions. %These results highlight the intuitive nature of using this novel approach to derive mean field state transition models from an explicitly stochastic assumptions framed in the context of Poisson Processes and CTMCs. %These results also summarize some of the computational tools available for constructing such models, and the relative ease of solving them numerically using standard numerical methods.

\subsection{Continuous Time Markov Chains and Phase-Type Distributions}

To provide proper context for an intuitive understanding of the Generalized Linear Chain Trick (GLCT), we briefly review continuous time Markov chains (CTMCs) with a focus on CTMCs that have a single absorbing state. Our focus will then be on the probability distributions that describe the time it takes to reach that absorbing state starting from one of the transient states, since these absorption time distributions define the phase-type family of probability distributions. The following summaries build upon similar descriptions laid out in \citet{Hurtado2019}.

\subsubsection{Continuous Time Markov Chains}

Discrete time Markov chains (DTMCs) describe the transition of an individual (or other distinct entity) among a set of $n$ states. The transition probabilities from a state $i$ to a state $j$ ($p_{ij}$ where $1\leq i,j\leq n$) are best organized using a transition probability matrix $\mathbf{P}$, where $p_{ij}$ is value in the $i^\text{th}$ row and $j^\text{th}$ column of the matrix ($P_{ij}=p_{ij}$).  For our purposes below, we will restrict our attention to Markov Chains in which the first $k=n-1$ states are \textit{transient states}, and the last $(k+1)$ state is an \textit{absorbing state}. This means the system eventually enters this last state and remains there on each subsequent time step with probability 1. 

The transition probability matrix $\mathbf{P}$ can be written in block form according to these first $k=n-1$ transient states (we'll call this set of states X) and the last absorbing state as  \begin{equation}
\mathbf{P} = \begin{bmatrix} \mathbf{P_X} & \mathbf{P_a} \\ \mathbf{0} & 1 \end{bmatrix}
\end{equation} where $\mathbf{P_X}$ is a $k\times k$ matrix describing transition probabilities among transient states, $\mathbf{P_a}$ is the $k\times 1$ vector of probabilities of transitioning from the $i^\text{th}$ transient state to the absorbing state, and $\mathbf{0}$ is a $1\times k$ vector of zeros. 

In a continuous time Markov chain (CTMC), these transitions don't occur according to a fixed time step, but instead each transition occurs after an exponentially distributed amount of time. If the individual is in state $i$, that exponential distribution has rate $\lambda_i$ or equivalently has mean duration $1/\lambda_i$. Let  $\mathbf{R}$ be the vector of the rates $\lambda_i$, for $i=1,\ldots,k$.  Due to the memorylessness property of exponential distributions transitions from state $i$ to state $i$ can effectively be ignored, thus we can formulate a new transition probability matrix that describes an equivalent Markov chain but where we only track transitions to new states. This reformulated Markov chain is known as the \textit{embedded jump process} (or \textit{embedded Markov chain}), and it is described with a transition probability matrix $\mathbf{\widetilde P}$ where the diagonal entries are 0 and the off-diagonal transition probabilities $\widetilde p_{ij} = p_{ij}/\sum_{j\neq i}p_{ij}$. That is, these are just the off-diagonal entries of $\mathbf{P}$ with the diagonal set to 0, and the rows normalized so each row sums to 1.  For our purposes, since a Markov chain with an absorbing state is not ergodic and therefore does not properly have an embedded jump process representation, the above procedure is only applied to $k$ rows corresponding to transient states. Thus, the last row of $\mathbf{P}$ and $\mathbf{\widetilde P}$ will be the same, so that the last state in the Markov chain remains an absorbing state. 

In a CTMC context, the transition probability matrix $\mathbf{\widetilde P}$ and rate vector $\mathbf{R}$ are combined into a single matrix called the \textit{transition rate matrix} $\mathbf{Q}$. The entries of $\mathbf{Q}$ can be thought of as the mean-field, per-individual loss rates from each state (along the diagonal) and the transition rates from state $i$ into state $j$ (the off diagonal entries; see below). It has the block form \begin{equation}
\mathbf{Q} = \begin{bmatrix} \mathbf{A} & \mathbf{B} \\ \mathbf{0} & 0 \end{bmatrix}
\end{equation}

where $\mathbf{A}$ is a $k\times k$ matrix describing the transition rates among transient states, $\mathbf{B}$ is the $k\times 1$ vector of transition rates from the transient states to the absorbing state, and the bottom row is all zeros for the absorbing state. 

Matrices $\mathbf{A}$ and $\mathbf{B}$ are constructed from the entries of the transition probability matrix $\mathbf{\widetilde P}$ and the vector of rates for the exponential distributed dwell-times $\mathbf{R}$ as follows. The $i^\text{th}$ diagonal entry of $\mathbf{A}$ is the loss rate $-\lambda_i$ and the rest of the entries in the $i^\text{th}$ row are the product of $\lambda_i$ and the transition probability $\widetilde p_{ij}$. That is, the $ij$ off-diagonal entries of $\mathbf{A}$ are the per-individual transition rates from the $i^\text{th}$ state into the $j^\text{th}$ state, given by $\lambda_i\,\widetilde p_{ij}$. Since the rows of $\mathbf{\widetilde P}$ sum to 1, the rows of $\mathbf{Q}$ sum to 0, and it then follows that vector $\mathbf{B}$ is equal to the negative of the row sums of $\mathbf{A}$. Thus, we can write $\mathbf{B} = -\mathbf{A}\,\mathbf{1}$, where $\mathbf{1}$ is a column vector of $k$ ones. 

Finally, assume that the initial state of such a CTMC is one of the $k$ transient states. Let $\boldsymbol{\alpha}$ be the length $k$ column vector of probabilities (i.e., $\sum_{i=1}^k \alpha_i = 1$) that define the initial state distribution over these $k$ states.

Note that, for CTMCs which have $k$ transient states and $1$ absorbing state, all of the information necessary to describe the CTMC is contained in the transition rate matrix for transient states, $\mathbf{A}$, and initial distribution vector $\boldsymbol{\alpha}$. As detailed next, these quantities are also sufficient to parameterize the corresponding phase-type distribution.

\subsubsection{Phase-Type Distributions}

With the above family of CTMCs in mind, let $T_i$ be defined as the duration of time that it takes to first reach the absorbing state, given the CTMC starts in the $i^\text{th}$ transient state. We call this an \textit{absorption time}. More generally, let $T$ be the absorption time given that the initial state is determined by the initial probability vector $\boldsymbol{\alpha}$ (i.e., $T$ follows the mixture distribution of random variables $T_i$ with mixing probabilities $\boldsymbol{\alpha}$). Phase-type distributions are the family of absorption time distributions for all such $T$ described above.

The most familiar examples are the exponential distribution, and the Erlang distribution (i.e., those gamma distributions that have an integer shape parameter $k\geq1$) which can be thought of as the sum of $k$ independent exponentially distributed random variables, each with rate $r$. Erlang distributions can be parameterized in terms of their mean $\mu$ and coefficient of variation $c_v$ (the standard deviation over the mean), or their rate $r$ and shape $k$, where \begin{equation} 
\mu=\frac{k}r,\quad \sigma^2 = \frac{k}{r^2}, \quad c_v = \frac{1}{\sqrt{k}}, \quad \text{and thus,} \quad r=\frac{\mu}{\sigma^2},\quad \text{and} \quad k=\frac{\mu^2}{\sigma^2}=\frac{1}{c_v^{\;2}}.
\end{equation}

More generally, phase-type distributions are parameterized by vector $\boldsymbol{\alpha}$ and transient state rate matrix $\mathbf{A}$ (as defined in the previous section), and have the probability density function, cumulative distribution function, and $j^\text{th}$ moment given (respectively) by:  \begin{subequations} \begin{align}
	f(t) =&\; \boldsymbol{\alpha}\,e^{\mathbf{A}t}\,(-\mathbf{A}\mathbf{1}) \\
	F(t) =&\; 1 - \boldsymbol{\alpha}\,e^{\mathbf{A}t}\,\mathbf{1} \\
	E(T^j)=&\; j!\,\boldsymbol{\alpha}\,(-\mathbf{A})^{-j}\mathbf{1}. 
\end{align} \end{subequations}

Here $\mathbf{1}$ is a column vector of ones that has the same number of rows as $\boldsymbol{\alpha}$ and $\mathbf{A}$. Note that $\boldsymbol{\alpha}$ and $\mathbf{A}$ are not a unique parameterization of a given phase-type distribution, and there are equivalent parameterizations using vector-matrix pairs of the same dimension as well as different dimensions. Phase-type distributions can be classified as cyclic (transient states can be visited infintely often) and acyclic (transient states can only be visited once). This family of distributions has been relatively well studied in the queuing theory literature, and elsewhere, and readers are encouraged to consult \citet{Bladt2017, Bladt2017ch3, Reinecke2012a, Reinecke2012b, Horvath2012, Horvath2016, Altiok1985} for further details.

Additionally, freely available computational tools such as BuTools for Matlab and Python \citep{BuTools2,BuToolswww} enable researchers to fit phase-type distributions to data. This fact, combined with the Generalized Linear Chain Trick, allows for the construction of ODE models that incorporate empirically derived distributional assumptions for the time spent in a given state.

\subsection{Generalized Linear Chain Trick}

The GLCT provides modelers with a direct way to take an existing ODE model that includes a state that has an exponentially distributed dwell time, and obtain a new set of ODEs where that exponentially distributed dwell time has been replaced with a phase-type dwell time distribution. This is done by partitioning that focal state into a set of sub-states and using the GLCT to write the new systems of ODEs that govern those sub-states using the matrix and vector parameterization of the assumed phase-type distribution. This technique can also be used to implement the classic Linear Chain Trick (LCT), since Erlang distributions (i.e., gamma distributions with integer shape parameters) are a subfamily of phase-type distributions. 

The GLCT in its most general form \citep{Hurtado2019} extends the GLCT for phase-type distributions to the scenario where the rates and probabilities in the CTMC framework described above can vary with time. Here, we only provide a statement of the GLCT for phase-type distributions:

\begin{theorem}[{GLCT for phase-type distributions [Corollary 2 in \cite{Hurtado2019}]}]
	Assume individuals enter state X at rate $\mathcal{I}(t)$ and that the distribution of time spent in state X follows a continuous phase-type distribution given by the length $k$ initial probability vector $\alpha$ and the $k\times k$ matrix $A$. The mean field equations for these transient sub-states $x_i$ are given by
	\begin{equation} \frac{d}{dt}\mathbf{x}(t)=\boldsymbol\alpha\,\mathcal{I}(t) + \mathbf{A}^\text{T}\,\mathbf{x}(t) \label{eq:GLCT}\end{equation}
	where the rate of individuals leaving each of these sub-states of X is given by the vector $(-\mathbf{A\,1})\circ\mathbf{x}$, where $\circ$ is the Hadamard (element-wise) product of the two vectors, and thus the total rate of individuals leaving state X is given by the sum of those terms, i.e., $(-\mathbf{A\,1})^\text{T}\mathbf{x}=-\mathbf{1}^\text{T}\mathbf{A}^\text{T}\mathbf{x}$.
\end{theorem}

Note that the influx of individuals at time $t$ (at rate $\mathcal{I}(t)$) is distributed across the sub-states of X according to the initial distribution vector $\boldsymbol{\alpha}$, and the second term in eq. \eqref{eq:GLCT} describes both the movements among sub-states of X as well as the loss rate from the state X from each sub-state.

The Linear Chain Trick (LCT) has been known for decades, and is special case of the GLCT for phase-type distributions stated above \citep{Hurtado2019}. Here we give a formal statement of the LCT, which assumes an Erlang distributed dwell time, with shape parameter $k$ and rate parameter $r$. 

\begin{corollary}[\textbf{Linear Chain Trick}] \label{th:lct}
	~\\
	Consider the GLCT above. Assume that the dwell-time distribution is an Erlang distribution with shape $k$ and rate $r$ (or if written in terms of shape $k$ and mean $\tau=k/r$, then use rate $r=k/\tau$). Then the corresponding mean field ODE equations for the $k$ sub-states of X are 
	\begin{equation} \label{eq:LCT}\begin{split} 
	\frac{dx_1}{dt} =&\; \mathcal{I}(t) - r\,x_1 \\
	\frac{dx_2}{dt} =&\;  r\,x_1 - r\,x_2 \\
	&\vdots \\
	\frac{dx_k}{dt} =&\; r\,x_{k-1} - r\,x_{k}.
	\end{split}\end{equation} where the total loss rate from state X at time $t$ is the loss rate from the final sub-state, $r\,x_k(t)$.
\end{corollary}
\begin{proof}
	The phase-type distribution formulation of the Erlang distribution with shape $k$ and rate $r$ is given by \textbf{v} and \textbf{M} below, and substituting these into eq. \eqref{eq:GLCT} which yields the desired result. \begin{equation} \label{eq:vM} \textnormal{\textbf{v}} = \begin{bmatrix} 1 \\ 0 \\ \vdots \\ 0  \end{bmatrix}  \qquad \text{and} \qquad \mathbf{M} = \begin{bmatrix} -r & r & 0 & \cdots & 0 \\
	0 & -r & r  & \ddots & 0 \\
	\vdots & \ddots & \ddots  & \ddots & \ddots \\
	0 & 0 & \ddots &  -r & r \\
	0 & 0 & \cdots &  0 & -r
	\end{bmatrix} \end{equation}  See \citet{Hurtado2019} for a direct proof that uses a recursive relationship between Erlang density functions and their derivatives.
\end{proof}

\section{Results}\label{sec:results}

In the sections below, we extend two well-known models using the GLCT by replacing the implicit exponentially distributed dwell time assumptions of these models with phase-type distribution assumptions. These more general model formulations can also be used as a way to formulate models that could otherwise be derived using the standard LCT (i.e., the assumption of Erlang distributed dwell times). This may be the more desirable approach since the phase-type formulation of such models can be more practically and computationally advantageous to work with, which we show in section \ref{sec:benchmark}.

\subsection{Rosenzweig-MacArthur Predator-Prey Model}

Maturation delays in population models can influence model outputs, although such delays are not always incorporated into models used in applications \citep{Robertson2018}. In this section, we illustrate how one can use the GLCT to incorporate phase-type maturation times into such population models, using the widely used Rosenzweig-MacArthur model of predator-prey (consumer-resource) dynamics \citep{Murdoch2003,Rosenzweig1963}: \begin{subequations} \label{eq:RM} \begin{align} 
	\frac{dN}{dt} =&\; r\,N\bigg(1-\frac{N}{K}\bigg)-\frac{a\,P}{k+N}N\\
	\frac{dP}{dt} =&\; \chi\frac{a\,N}{k+N}P - \mu\,P  
\end{align} \end{subequations} 

In the absence of predators ($P$), the prey population ($N$) is subject to logistic growth, and predators consume prey following a Holling's type II functional response \citep{Murdoch2003, Dawes2013, Holling1959, Holling1959a}. Predators will then live for an exponentially distributed lifetime with mean $1/\mu$. 

One approach to incorporating a maturation delay of duration $\tau$ is to formulate a delay differential equation (DDE), as in \citep{Xia2009}: \begin{subequations} \label{eq:RMDDE} \begin{align} 
	\frac{dN(t)}{dt} =&\; r\,N(t)\bigg(1-\frac{N(t)}{K}\bigg)-\frac{a\,P(t)}{k+N(t)}N(t)\\
	\frac{dP(t)}{dt} =&\; \chi\frac{a\,N(t-\tau)}{k+N(t-\tau)}P(t-\tau) - \mu\,P(t)  
\end{align} \end{subequations} 

This can be thought of as the limit of a distributed delay model, with mean delay time $\tau$, for which the variance or coefficient of variation goes to zero. This corresponds to a delay distribution with point mass at $\tau$ (i.e., the distribution can be described with a Dirac delta function). The LCT has long been used to approximate such limiting cases in DDE models by assuming instead a delay distribution that is Erlang distributed with mean $\tau$ and a very small coefficient of variation, i.e., a large shape parameter $k=1/c_v^{\;2}$ \citep{Smith2010, Hurtado2019}. Writing this approximating model, as in \citet{Hurtado2020}, yields the Rosenzweig\textendash{}MacArthur model with Erlang distributed maturation time in the predators: \begin{subequations} \label{eq:RMLCT} \begin{align} 
	\frac{dN}{dt} =&\;  r\,N\bigg(1-\frac{N}{K}\bigg)-\frac{a\,P}{k+N}N\\
	\frac{dx_1}{dt}=&\; \chi\frac{a\,N}{k+N}P - \frac{k}{\tau}\,x_1 \\
	\frac{dx_j}{dt}=&\; \frac{k}{\tau}\,x_{j-1} - \frac{k}{\tau}\,x_j, \quad \text{for } j=2,\ldots,k\\
	\frac{dP}{dt} =&\;  \frac{k}{\tau}\,x_k - \mu\,P  
\end{align} \end{subequations}  

The sub-states $x_j$, $j=1,\ldots,k$, track the immature stages of the predators before they mature. 

Using the GLCT, the above model can be generalized in two ways. First, the Erlang distributed maturation time assumption that yields the sub-states $x_i$ can be replaced by the assumption of a more general phase-type distribution with matrix-vector parameterization $\boldsymbol\alpha_\mathbf{X}$ and $\mathbf{A_X}$. Similarly, the exponentially distributed time duration that predators spend as adults can also be replaced with a more general phase-type distribution with parameter vector $\boldsymbol\alpha_\mathbf{P}$ and matrix $\mathbf{A_P}$.  According to the GLCT \textendash{} where $\mathbf{x}(t)$ denotes the vector of maturing predator sub-states $x_i(t)$, $\mathbf{y}(t)$ is the vector of adult predator sub-states $y_j(t)$, and where $P(t)=\sum_{\text{all }j}y_j(t)$ \textendash{} these assumptions yield the more general model: \begin{subequations} \label{eq:RMPT} \begin{align} 
	\frac{dN}{dt} =&\; r\,N\bigg(1-\frac{N}{K}\bigg)-\frac{a\,P}{k+N}N\\
	\frac{d\mathbf{x}}{dt} =&\; \chi\frac{a\,N}{k+N}P\;\boldsymbol{\alpha}_\mathbf{X} + \mathbf{A_X}^\text{T} \mathbf{x}  \\
	\frac{d\mathbf{y}}{dt} =&\; \underbrace{-\mathbf{1}^\text{T}\mathbf{A_X}^\text{T}\mathbf{x}}_\text{scalar}\;\boldsymbol{\alpha}_\mathbf{P} + \mathbf{A_P}^\text{T}\mathbf{y}  
\end{align} \end{subequations} 

Observe that eqs. \eqref{eq:RMLCT} are the special case of eqs. \eqref{eq:RMPT} where the phase-type distribution matrix-vector parameterization for an Erlang distribution with mean $\tau$ and shape $k$ is given by \begin{equation}\boldsymbol{\alpha}_\mathbf{X} = \begin{bmatrix} 1 \\ 0 \\ \vdots \\ 0  \end{bmatrix} \quad \text{and} \quad  \mathbf{A_X} = \begin{bmatrix} -\frac{k}{\tau} & \frac{k}{\tau} & 0 & 0 & \cdots & 0 \\
0 & -\frac{k}{\tau} & \frac{k}{\tau} & 0 & \cdots & 0 \\
& & & \ddots & & \\
0 & 0 & 0 & \cdots & -\frac{k}{\tau}  & \frac{k}{\tau} \\
0 & 0 & 0 & \cdots & 0 & -\frac{k}{\tau} 
\end{bmatrix} \end{equation}

and for an exponential distribution with rate $\mu$, \begin{equation}\boldsymbol{\alpha}_\mathbf{P} = \begin{bmatrix} 1 \end{bmatrix} \quad \text{and} \quad  \mathbf{A_P} = \begin{bmatrix} -\mu \\
\end{bmatrix}. \end{equation}

Note that eqs. \eqref{eq:RMLCT} are a much more compact way of formulating such generalized models without the need to specify the number of sub-states. As shown below, this formulation allows modelers to write more efficient computer code for computing numerical solutions to such models, and also can be used with computer algebra systems to generate ODEs like eqs. \eqref{eq:RMLCT} from first principles.

% Parameters to use: IC = c(N=1000,P=10); params = c(r = 1, K = 1000, a = 5, k = 500, chi = 0.5, mu = 1)

\subsection{SEIR Model} \label{sec:SEIR}

SIR-type models of infectious disease transmission are widely used in the study of infectious diseases, and can help inform public health efforts to limit the spread of infectious diseases \citep{Kermack1927,Kermack1932,Kermack1933,Kermack1991a,Kermack1991b,Kermack1991c,Keeling1997,AndersonMay1992,Diekmann2000,Lloyd2009,Wearing2005}. For example, such models are currently being used in response to the ongoing SARS-CoV-2/COVID-19 pandemic. 

It is known that including a latent period prior to the onset of symptoms and infectiousness, as well as incorporating non-exponential distributions for the time spent in the different infection states, can be important to include in models that are being used in such applications \citep{Wearing2005,Feng2007,Wang2017}. 

Here we use the GLCT to formulate a more general SEIR model where we assume that the latent period (time spent in state E) and infectious period (time spent in state I) follow independent phase-type distributions. For simplicity, we assume the state variables have been scaled by the total population size so that $S+E+I+R=1$, and that there are no births or deaths in the model. 

To begin, consider this simple SEIR model, where $S$ is the fraction of susceptibles in the population, $E$ the fraction of exposed individuals with latent infections, $I$ the fraction of individuals with active infections, and $R$ the fraction of recovered or removed individuals: \begin{subequations}  \label{eq:SEIR}  \begin{align}
	\frac{dS}{dt} =& \; -\beta\,S\,I \label{eq:SEIRa} \\
	\frac{dE}{dt} =& \; \beta\,S\,I - r_E\,E \label{eq:SEIRb}\\
	\frac{dI}{dt} =& \; r_E\,E - r_I\,I \label{eq:SEIRc}\\
	\frac{dR}{dt} =& \; r_I\,I \label{eq:SEIRd} 
	\end{align}
\end{subequations}

Next, assume the latent period distribution is phase-type with parameters $\boldsymbol{\alpha}_\mathbf{E}$ and $\mathbf{A_E}$, and the infectious period distribution is also phase-type, but with parameters $\boldsymbol{\alpha}_\mathbf{I}$ and $\mathbf{A_I}$. Let $\mathbf{x}=[E_1,\ldots]^\text{T}$ and $\mathbf{y}=[I_1,\ldots]^\text{T}$ be the column vectors of the fraction of individuals in each of the exposed and infectious sub-states, respectively, where $E=\sum E_j$ and $I=\sum I_i$. Then by the GLCT we can write the mean field ODEs for the generalized SEIR model as follows: \begin{subequations} \label{eq:SEIRPT} \begin{align}  
\frac{dS}{dt} =& \; -\beta\,S\,I  \\
\frac{d\mathbf{x}}{dt} =& \; \boldsymbol{\alpha}_\mathbf{E}\,\beta\,S\,I + {\mathbf{A_E}}^\textsf{T}\mathbf{x} \\
\frac{d\mathbf{y}}{dt} =& \; \boldsymbol{\alpha}_\mathbf{I}\underbrace{\big((-{\mathbf{A_E}\mathbf{1}})^\textsf{T}\mathbf{x}\big)}_\text{scalars\qquad} + {\mathbf{A_I}}^\textsf{T}\mathbf{y} \\
\frac{dR}{dt} =& \; \overbrace{(-{\mathbf{A_I}\mathbf{1}})^\textsf{T}\mathbf{y}}
\end{align} \end{subequations} 
 
To assume, for example, an Erlang distributed latent period with mean $\tau_E$ and shape $k_E$, i.e., rate $1/\tau_E$ and coefficient of variation $c_v=1/\sqrt{k_E}$, then one would use \begin{equation}\boldsymbol{\alpha}_\mathbf{E} = \begin{bmatrix} 1 \\ 0 \\ \vdots \\ 0  \end{bmatrix} \quad \text{and} \quad  \mathbf{A_E} = \begin{bmatrix} -\frac{k_E}{\tau_E} & \frac{k_E}{\tau_E} & 0 & 0 & \cdots & 0 \\
0 & -\frac{k}{\tau_E} & \frac{k_E}{\tau_E} & 0 & \cdots & 0 \\
& & & \ddots & & \\
0 & 0 & 0 & \cdots & -\frac{k_E}{\tau_E}  & \frac{k_E}{\tau_E} \\
0 & 0 & 0 & \cdots & 0 & -\frac{k_E}{\tau_E} 
\end{bmatrix} \end{equation}

Similarly, an Erlang infectious period distribution with mean $\tau_I$ and shape parameter $k_I$ (coefficient of variation $1/\sqrt{k_I}$) would be parameterized by \begin{equation}\boldsymbol{\alpha}_\mathbf{I} = \begin{bmatrix} 1 \\ 0 \\ \vdots \\ 0  \end{bmatrix} \quad \text{and} \quad  \mathbf{A_I} = \begin{bmatrix} -\frac{k_I}{\tau_I} & \frac{k_I}{\tau_I} & 0 & 0 & \cdots & 0 \\
0 & -\frac{k}{\tau_I} & \frac{k_I}{\tau_I} & 0 & \cdots & 0 \\
& & & \ddots & & \\
0 & 0 & 0 & \cdots & -\frac{k_I}{\tau_I}  & \frac{k_I}{\tau_I} \\
0 & 0 & 0 & \cdots & 0 & -\frac{k_I}{\tau_I} 
\end{bmatrix} \end{equation}

Simplifying the right hand side of eqs. \eqref{eq:SEIRPT} using the matrix and vector definitions above yields the familiar sub-state equations for an SEIR model with Erlang distributed latent and infectious periods, eqs. \eqref{eq:SEIRLCT}. 

\begin{subequations} \label{eq:SEIRLCT} \begin{align}
\frac{dS}{dt} =& \; -\beta\,S\,I  \\
\frac{d\mathbf{x}}{dt} =&\; \begin{bmatrix} \beta\,SI -\frac{k_E}{\tau_E}E_1 \\ \frac{k_E}{\tau_E}E_1 - \frac{k_E}{\tau_E}E_2 \\ \frac{k_E}{\tau_E}E_2 - \frac{k_E}{\tau_E}E_3 \\ \vdots \\ \frac{k_E}{\tau_E}E_{k_E-1} - \frac{k_E}{\tau_E}E_{k_E} \end{bmatrix} \\
\frac{d\mathbf{y}}{dt} =&\; \begin{bmatrix} \frac{k_E}{\tau_E}E_{k_E} -\frac{k_I}{\tau_I}I_1 \\ \frac{k_I}{\tau_I}I_1 - \frac{k_I}{\tau_I}I_2 \\ \frac{k_I}{\tau_I}I_2 - \frac{k_I}{\tau_I}I_3 \\ \vdots \\ \frac{k_I}{\tau_I}I_{k_I-1} - \frac{k_I}{\tau_I}I_{k_I} \end{bmatrix} \\
\frac{dR}{dt} =& \; \frac{k_I}{\tau_I}I_{k_I} \end{align}
\end{subequations} where $\mathbf{x} = [E_1,\;\ldots,E_{k_E}]^\text{T}$, $\mathbf{y} = [I_1,\;\ldots,I_{k_I}]^\text{T}$, and $I=\sum I_i$.

Other phase-type distributions could be assumed, e.g., by fitting non-Erlang phase-type distributions to data using computational tools like the free software \texttt{BuTools} \citep{BuTools2,BuToolswww}.

\subsection{SEIR Model with Heterogeneity Among Infected Individuals}

The examples above illustrate how an existing DDE or ODE model can be generalized by assuming that states with fixed or exponentially distributed dwell times instead have phase-type distributed dwell times. Here we take the generalized SEIR model above and use the GLCT to further explore more complex model assumptions. We do this by considering two special cases of this generalized model (see Figs. \ref{fig:SEIRH1}, \ref{fig:SEIRH2}): one that models hospitalization in a manner that does not change the distribution of time in the infected class, and a second case that models heterogeneity in the severity and duration of disease. In each case, we assume that infected individuals will either experience mild or severe disease with the potential for distributional differences in the duration of infection.

%The GLCT enables us to formulate a very general model, and within that framework implement different model assumptions surrounding the progression through more severe illness and hospitalization. Importantly, such assumptions may relate to the distribution of time spent in certain states, as well as the model structure of how individuals pass through multiple states.

For simplicity, here we assume the removed class $R$ contains both recovered and deceased individuals, and that, upon transitioning from the class of individuals with latent infection (E) to the infectious class (I), a fraction $\rho$ will go on to develop severe symptoms (in state I$_s$) that may require hospitalization, whereas the remaining fraction ($1-\rho$) do not develop severe disease and instead enter a different state of more mild disease (I$_0$). Using the GLCT, the states I$_0$ and I$_s$ are partitioned into sub-states, where the numbers in each are described by vectors $\mathbf{y_{0}}$ and $\mathbf{y_{s}}$, respectively, and their dynamics obey the following system of ODEs: \begin{subequations} \label{eq:SEIRH} \begin{align}  
	\frac{dS}{dt} =& \; -\beta\,S\,I  \\
	\frac{d\mathbf{x}}{dt} =& \; \boldsymbol{\alpha}_\mathbf{E}\,\beta\,S\,I + {\mathbf{A_E}}^\textsf{T}\mathbf{x} \\
	\frac{d\mathbf{y_0}}{dt} =& \; \boldsymbol{\alpha}_\mathbf{I_0}\big((1-\rho)(-{\mathbf{A_E}\mathbf{1}})^\textsf{T}\mathbf{x}\big) + {\mathbf{A_{I_0}}}^\textsf{T}\mathbf{y_0} \\
	\frac{d\mathbf{y_s}}{dt} =& \; \boldsymbol{\alpha}_\mathbf{I_s}\big(\rho\,(-{\mathbf{A_E}\mathbf{1}})^\textsf{T}\mathbf{x}\big) + {\mathbf{A_{I_s}}}^\textsf{T}\mathbf{y_s} \\
	\frac{dR}{dt} =& \; (-{\mathbf{A_{I_0}}\mathbf{1}})^\textsf{T}\mathbf{y_0} +  (-{\mathbf{A_{I_s}}\mathbf{1}})^\textsf{T}\mathbf{y_s}
	\end{align} \end{subequations} 

Eqs. \eqref{eq:SEIRH} can be viewed as a special case of eqs. \eqref{eq:SEIRPT} defined in terms of a mixture of two phase-type distributions, where $\mathbf{y}=[\mathbf{y_0}^\text{T},\mathbf{y_s}^\text{T}]$, the vector $\boldsymbol{\alpha_I}=[(1-\rho)\boldsymbol{\alpha_{I_0}}^\text{T},\,\rho\boldsymbol{\alpha_{I_s}}^\text{T}]$, and $\mathbf{A_I}$ is the block diagonal matrix $\mathbf{A_I}=$diag($\mathbf{A_{I_0}},\mathbf{A_{I_s}}$). The two cases below are described in the context of eqs. \eqref{eq:SEIRH}.

\subsubsection{Case 1: Hospitalization Independent of Progress Towards Infection Resolution}

In this scenario, individuals progress towards recovery or death according to an Erlang distribution with rate $\lambda$ and shape $k_{y_0}$. Independently, they also move towards hospitalization according to an Erlang distributed hospitalization time with rate $r$ and shape $k_H$ \citep[for an example of this structure used to model Ebola, but where $k_H=1$, see][]{Feng2016}. Modeling this with Erlang latent and infectious period distributions for simplicity, and according to the competing Poisson processes motif detailed in \citep{Hurtado2019}, yields a sub-state structure as shown in Fig. \ref{fig:SEIRH1}.   The matrix-vector pairs $\boldsymbol{\alpha}_\mathbf{E}$, $\mathbf{A_E}$, and $\boldsymbol{\alpha}_\mathbf{I_0}$, $\mathbf{A_{I_0}}$ are as described above for Erlang distributions. 

The matrix-vector pair $\mathbf{A_{I_s}}$ and $\boldsymbol{\alpha}_\mathbf{I_s}$ are defined as follows\footnote{Compare to the matrix-vector parameterization of the minimum of two Erlang distributions (the minimum of two phase-type distributions is itself phase-type) using the formulas given in \citet{Bladt2017ch3}.}. If we order these states starting at the $I_{11}$ entry and work across each rows left to right before moving down to the next row (see Fig. \ref{fig:SEIRH1}), then the associated rate matrix has the following block form with each column and row having $k_H$ blocks of dimension $k_1 \times k_1$. This block structure corresponds to each row of the I$_1$ sub-states shown in the lower portion of Fig. \ref{fig:SEIRH1} as a $k_H\times k_1$ grid of sub-states. \begin{equation}
\mathbf{A_{I_s}}=\left[  \begin{array}{cccccc} %\begin{array}{c|c|c|c|c|c} 
\mathbf{A_{d1}} & \mathbf{A_{sup}} &  \mathbf{0} &  \mathbf{0} &  \cdots &  \mathbf{0}  \\
%\hline % ---------------------------------------------------------------------------------------------------
\mathbf{0} &\mathbf{A_{d1}} & \mathbf{A_{sup}} &  \mathbf{0}  &  \cdots & \mathbf{0} \\
%\hline % ---------------------------------------------------------------------------------------------------
 \mathbf{0}  & \mathbf{0} &\mathbf{A_{d1}} & \mathbf{A_{sup}} &  \ddots & \mathbf{0} \\
%\hline %-----------------------------------------------------------------------------------------------------
\vdots & \ddots & \ddots & \ddots & \ddots & \vdots \\
%\hline %-----------------------------------------------------------------------------------------------------
\mathbf{0} & \mathbf{0} & \mathbf{0} &   \mathbf{0}  & \mathbf{A_{d1}} & \mathbf{A_{sup}}  \\
%\hline % ----------------------------------------------------------------------------------------------------
\mathbf{0} & \mathbf{0} &   \cdots &  \mathbf{0}  &  \mathbf{0}  & \mathbf{A_{d2}} \\
\end{array} \right]
\end{equation}  where the diagonal and superdiagonal blocks are  \begin{equation}
	\mathbf{A_{d1}} = \begin{bmatrix} 
	-\lambda-r & \lambda & 0 & \cdots & 0  \\
	0 & \ddots & \lambda & \cdots & 0 \\
	\vdots & & \ddots & \ddots &  \vdots  \\
	0 & 0 & \cdots & \ddots & \lambda \\
	0 & 0 & \cdots & 0 & -\lambda-r  \\  \end{bmatrix}, \qquad \mathbf{A_{sup}}= \begin{bmatrix} 
	\;\;\;\;r\;\;\;\; & 0 & 0 & \cdots & 0  \\
	0 & \ddots & 0 & \cdots & 0 \\
	\vdots & & \ddots & \ddots &  \vdots  \\
	0 & 0 & \cdots & \ddots & 0\\
	0 & 0 & \cdots & 0 & \;\;\;\;r\;\;\;\;  \\  \end{bmatrix} 
\end{equation} and the bottom right diagonal entry is  \begin{equation}
	\mathbf{A_{d2}} = \begin{bmatrix} 
		\;\;-\lambda\;\;\; & \lambda & 0 & \cdots & 0  \\
		0 & \ddots & \lambda & \cdots & 0 \\
		\vdots & & \ddots & \ddots &  \vdots  \\
		0 & 0 & \cdots & \ddots & \lambda \\
		0 & 0 & \cdots & 0 & \;\;\;-\lambda\;\; \\  \end{bmatrix}.
\end{equation}

\begin{figure}[tb!] 
	\includegraphics[width=\textwidth]{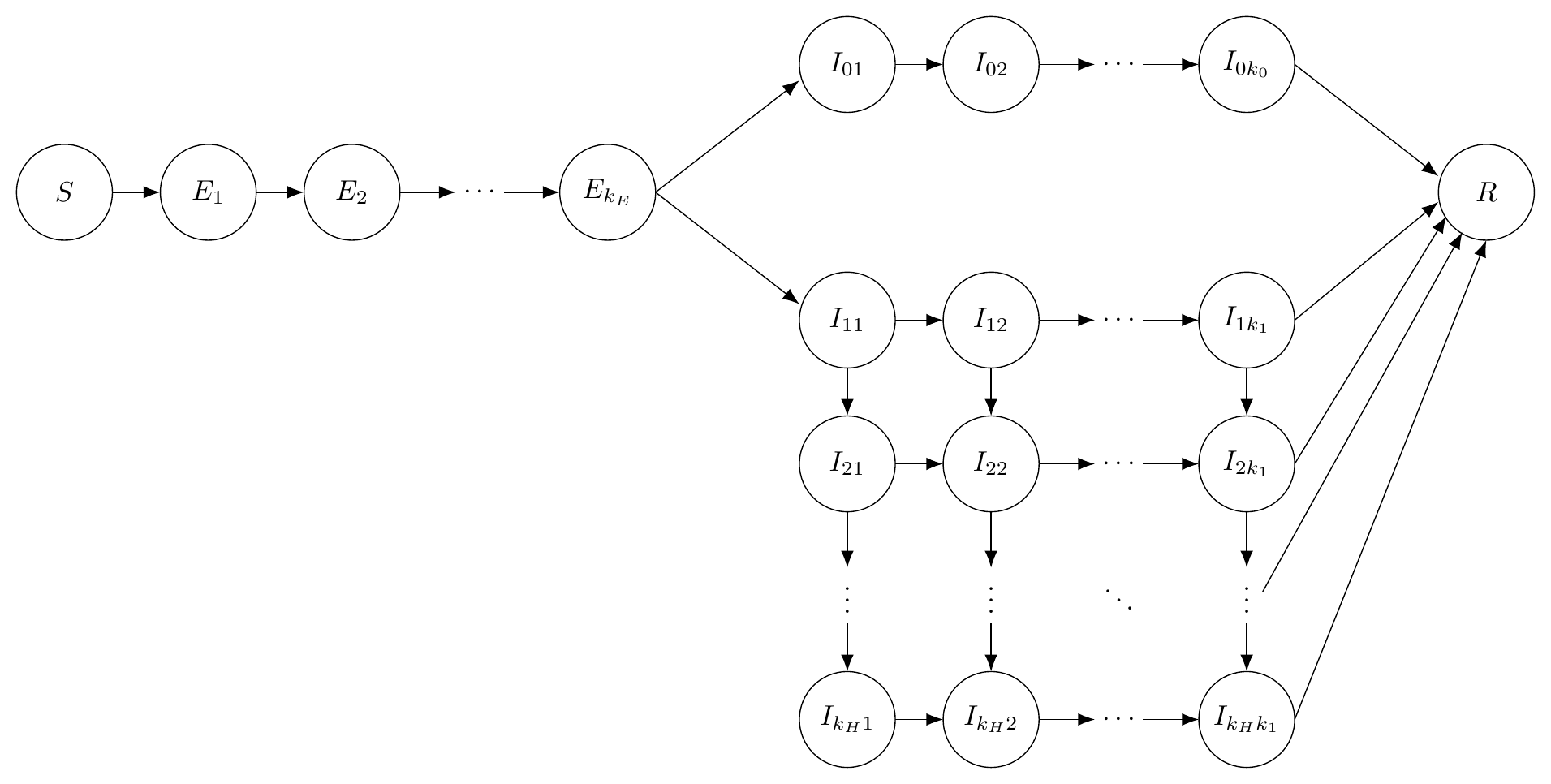}
	\caption{SEIR-type model with heterogeneity in illness severity and hospitalizations that do not alter the infectious period. A special case of eqs. \eqref{eq:SEIRH}  (cf. Fig. \ref{fig:SEIRH2}). See the main text for details. Here the standard LCT has been applied to the exposed (E) state, and the GLCT is applied to the infectious (I) states, using the competing Poisson Process approach \citep{Hurtado2019} to model hospitalizations in a fraction of the infectious individuals. This ensures that the time spent in I is independent of whether or not individuals transition to the hospital. }  \label{fig:SEIRH1}
\end{figure}

Together, the matrix $\mathbf{A_{I_s}}$ and the length $k_H\,k_{I_s}$ initial distribution vector $\boldsymbol{\alpha_{I_s}}=[1, \; 0, \; \cdots ,\; 0]^\text{T}$ complete the parameterization of model eqs. \eqref{eq:SEIRH} to yield the model structure illustrated in Fig. \ref{fig:SEIRH1}. 

Observe that the matrix above has the same diagonal and superdiagonal blocks in all rows except for the last row, for which the diagonal entries are $-\lambda$ and not $-\lambda-r$. Here the dwell time in each sub-state of I$_s$ (except for the last row) follows the minimum of two independent exponential distributions with rates $\lambda$ and $r$, so by the properties of exponential distributions, those dwell times are each exponentially distributed with rate $\lambda + r$. Individuals leaving those sub-states then either move horizontally towards resolving their infections with probability $\lambda/(\lambda+r)$, or move downwards towards hospitalization (the last row) with probability $r/(\lambda+r)$ \citep{Hurtado2019}. In the last row, individuals are hospitalized and can only move horizontally towards the resolution of their illness. This structure ensures that time transition to the hospital (the last row) does not impact their overall distribution of time spent in state I$_1$.

\subsubsection{Case 2: Hospitalization With Heterogeneous Need for Critical Care}

To further illustrate the flexibility of eqs. \eqref{eq:SEIRPT} and \eqref{eq:SEIRH}, we now consider the model structure illustrated in Fig. \ref{fig:SEIRH2}. In this case, we make similar assumptions to the case above, except for the states that pertain to the fraction $\rho$ of individuals who experience severe illness. Those individuals exhibit an Erlang distributed period of more mild disease, with rate $r_1$ and shape parameter $k_1$. Those individuals then either recover (with probability $f$) after an Erlang distributed period of time with rate $r_R$ and shape $k_R$, or they become even more ill (with probability $1-f$) and require hospitalization for an Erlang distributed amount time with rate $r_C$ and shape $k_C$. As above, all individuals eventually enter a \textit{removed} state R which, for our purposes here, makes no distinction between recovery and death. 

\begin{figure}[tb!] 
	\includegraphics[width=\textwidth]{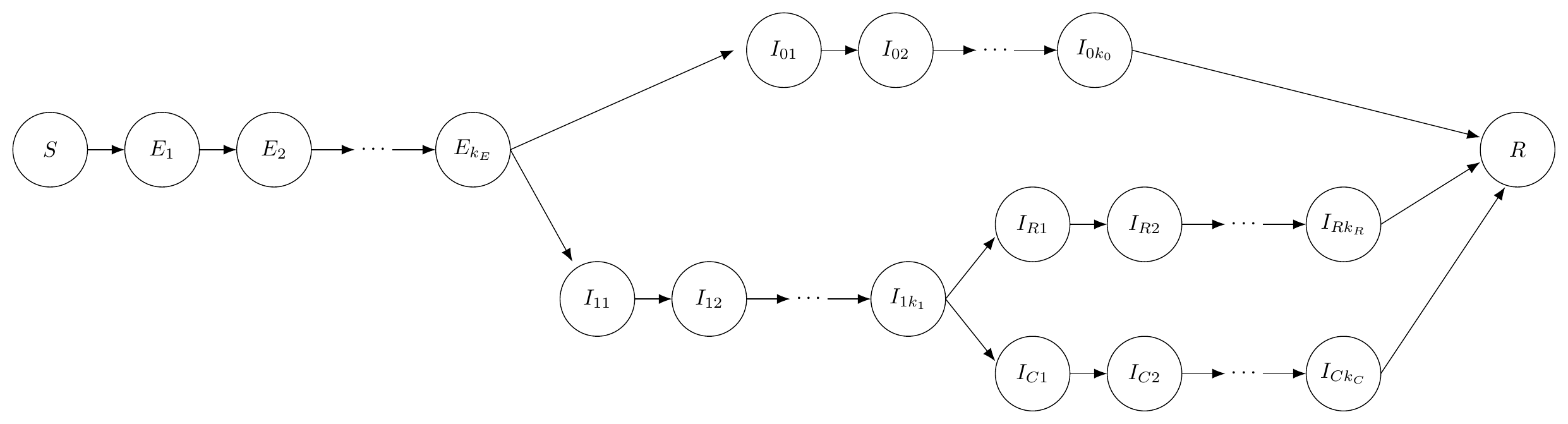}
	\caption{SEIR-type model with heterogeneity in illness severity in which a fraction of infected individuals experience severe illness and a fraction of those require critical care. This is also a special case of eqs. \eqref{eq:SEIRH} (see Fig. \ref{fig:SEIRH1}). See the main text for details.}  \label{fig:SEIRH2}
\end{figure}

Comparing Figs. \ref{fig:SEIRH1} and \ref{fig:SEIRH2}, we see that this second scenario is also a special case of eqs. \eqref{eq:SEIRH}, and only differs from that case in the definition of matrix $\mathbf{A_{I_s}}$ and the length $k_1+k_R+k_C$ initial distribution vector $\boldsymbol{\alpha_{I_s}}=[1, \; 0, \; \cdots ,\; 0]^\text{T}$. Ordering the substates of I$_s$ from I$_{11}$ to I$_{1k_1}$ to I$_{R1}$ to I$_{Rk_R}$ to I$_{C1}$ to I$_{Ck_C}$, then, by the assumptions above, $\mathbf{A_{I_s}}$ is given by \begin{equation}	\mathbf{A_{I_s}}=
\left[\footnotesize  \begin{array}{c|c|c} 
\begin{matrix} 
-r_1 & r_1 & 0 & \cdots & 0  \\
0 & -r_1 & r_1 & \cdots & 0 \\
\vdots & & \ddots & \ddots &  \vdots  \\
0 & 0 & \cdots & r_1 & r_1 \\
0 & 0 & \cdots & 0 & -r_1  \\  \end{matrix} & \begin{matrix} 
0 & 0 & 0 & \cdots & 0  \\
0 & \ddots & 0 & \cdots & 0 \\
\vdots & & \ddots & \ddots &  \vdots  \\
0 & 0 & \cdots & \ddots & 0\\
f\,r_1 & 0 & \cdots & 0 & 0  \\  \end{matrix}  &   \begin{matrix} 
0 & 0 & 0 & \cdots & 0  \\
0 & \ddots & 0 & \cdots & 0 \\
\vdots & & \ddots & \ddots &  \vdots  \\
0 & 0 & \cdots & \ddots & 0\\
(1-f)\,r_1 & 0 & \cdots & 0 & 0  \\  \end{matrix}  \\
\hline % ---------------------------------------------------------------------------------------------------
\mathbf{0} & \begin{matrix} 
-r_R & r_R & 0 & \cdots & 0  \\
0 & -r_R & r_R & \cdots & 0 \\
\vdots & & \ddots & \ddots &  \vdots  \\
0 & 0 & \cdots & -r_R & r_R \\
0 & 0 & \cdots & 0 & -r_R  \\  \end{matrix} &  \mathbf{0}  \\
\hline %-----------------------------------------------------------------------------------------------------
\mathbf{0} & \mathbf{0} & \begin{matrix} 
-r_C & r_C & 0 & \cdots & 0  \\
0 & -r_C & r_C & \cdots & 0 \\
\vdots & & \ddots & \ddots &  \vdots  \\
0 & 0 & \cdots & -r_C & r_C \\
0 & 0 & \cdots & 0 & -r_C  \\  \end{matrix} \\
\end{array} \right].
\end{equation}

The two examples above illustrate the utility of deriving models using the GLCT and by thinking about deriving model structure from first principles using intuition from CTMCs. This approach can be leveraged for the analytical study of such models (Hurtado and Richards, \textit{in prep.}), but as we show in the next section it can also facilitate the process of computing numerical solutions to such systems of ODEs. % REFS CITATION FIX THIS CITE ARXIV PAPER (())

\subsection{Benchmarking Numerical Solutions: LCT vs GLCT}\label{sec:benchmark}

Software like Matlab, Python, Julia, and R, have built-in ODE solvers that implement various numerical methods for obtaining numerical solutions to ODEs. In Fig. \ref{fig:RMbench}, we summarize the average time it takes to compute a numerical solution to the generalized Rosenzweig-MacArthur model with Erlang distributed maturation time and time predators spend in the adult stage, either in the (LCT) form of eqs. \eqref{eq:RMLCT} or in the (GLCT) form of eq. \eqref{eq:RMPT}. In Fig. \ref{fig:SEIRbench}, we make a similar comparison with the SEIR model with Erlang latent and infectious periods comparing the time it takes to compute numerical solutions to that model in the form of eqs. \eqref{eq:SEIRLCT} or the equivalent model in the (GLCT) form of eqs.  \eqref{eq:SEIRPT}. See the figure captions, the R code in Appendix A, and the Electronic Supplements for further computational details.

\begin{figure}[tbh!]
	\includegraphics[width=\textwidth]{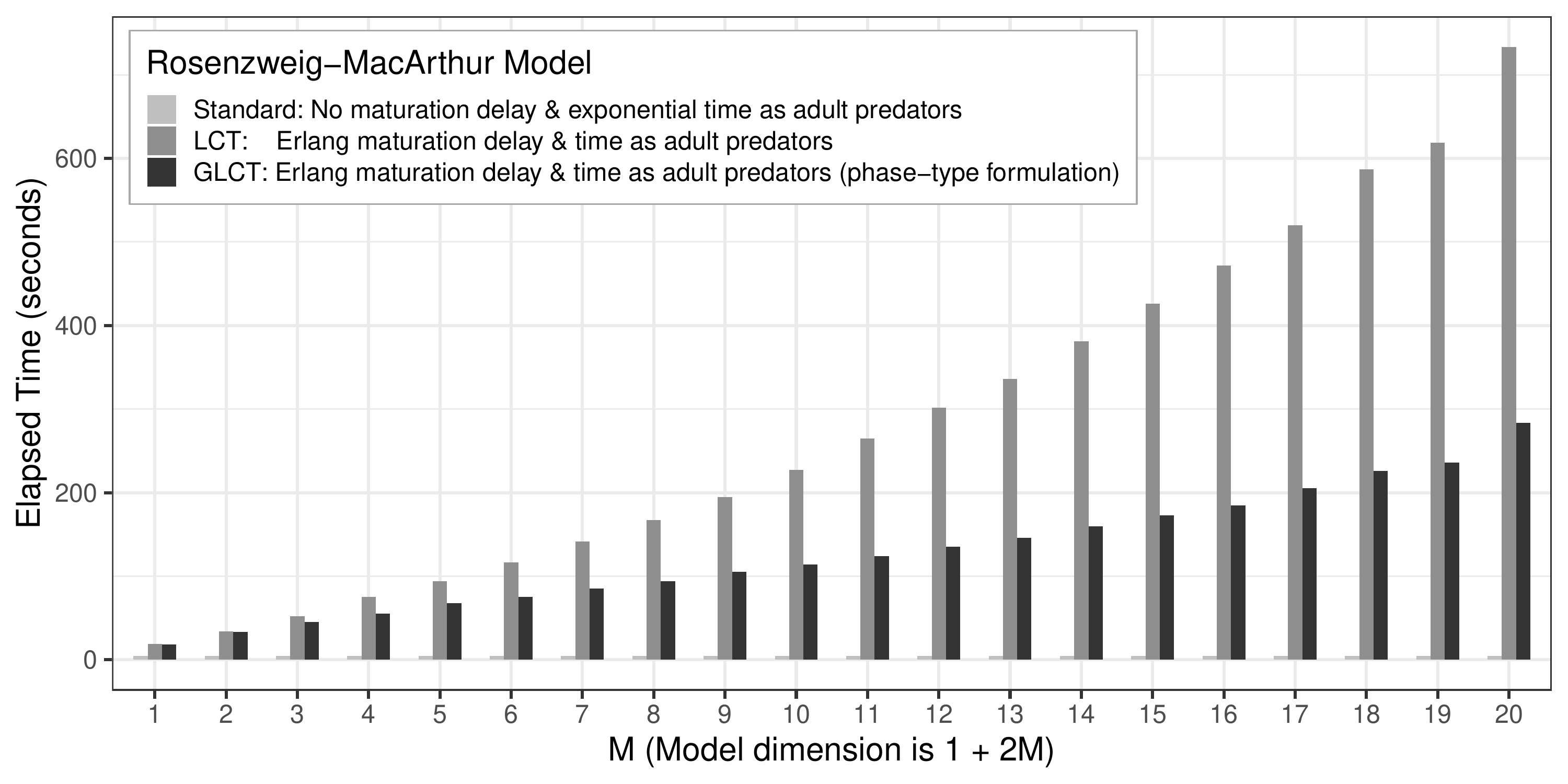}
	\caption{Benchmark results for 140 iterations of computing numerical solutions to the Rosenzweig-MacArthur model with Erlang (gamma) distributed maturation times and time spent in the adult-stage, using either a direct (LCT; medium gray) or more general (GLCT; black) model formulations (the standard Rosenzweig-MacArthur model with no maturation delay and exponentially distributed time spent in the adult stage [light gray; eqs. \eqref{eq:RM}] is included as a baseline). The second and third cases are mathematically equivalent systems. For smaller shape parameters (lower dimension systems) the GLCT model formulation is relatively slower than explicitly writing out the $2M+1$ equations, whereas for larger shape parameters (higher dimension systems) the GLCT formulation becomes markedly faster. This is likely due to the efficiency of the matrix computations. The x-axis values $M$ indicate the number of maturing predator sub-states ($k_x=M$) and adult predator sub-states ($k_y=M$), which yields a $2M+1$ dimensional model. Numerical solutions were computed using the \texttt{ode()} function in the \texttt{deSolve} package \citep{deSolve} in R \citep{R}, using method \texttt{ode45} with \texttt{atol=10\string^-6}, for time points $0$ to $500$ in increments of 1, with $r = 1$, $K = 1000$, $a = 5$, $k = 500$, $\chi = 0.5$, $\mu_x = 0.5$, $\mu_y = 1$, $N(0)=1000$, $x_i(0)=0$ ($i\geq1$), $y_1(0)=10$, and $y_j(0)=0$ ($j>1$). See Appendix \ref{sec:code} for details. } \label{fig:RMbench}
\end{figure}

\begin{figure}[tbh!]
	\includegraphics[width=\textwidth]{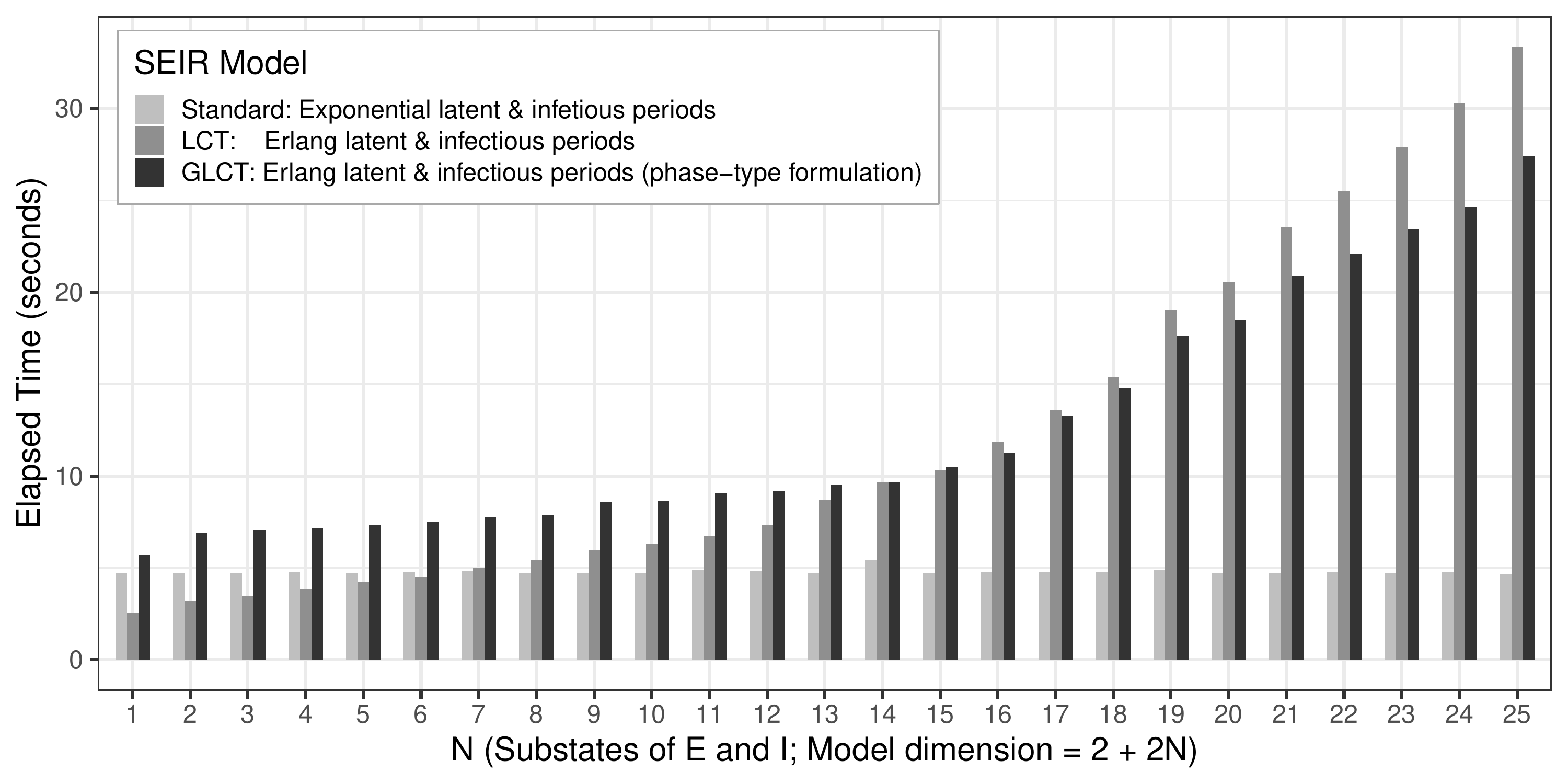}
	\caption{Benchmark results for 500 iterations of computing numerical solutions to the SEIR model with Erlang distributed latent and infectious periods, using either a direct (LCT; medium gray) or more general (GLCT; black) model formulations (SEIR with exponentially distributed latent and infectious periods [light gray; eqs. \eqref{eq:SEIR}] is included as a baseline). The second and third cases are mathematically equivalent systems. For smaller shape parameters (lower dimension systems) the GLCT model formulation is relatively slower than explicitly writing out the 2N+2 equations, whereas for larger shape parameters (higher dimension systems) the GLCT formulation becomes faster, likely due to the efficiency of the matrix computations. The x-axis values $N$ indicate the number of sub-states in each of E ($k_E=N$) and I ($k_I=N$), which yields a $2N+2$ dimensional model. Numerical solutions were computed using the \texttt{ode()} function in the \texttt{deSolve} package \citep{deSolve} in R \citep{R}, using method \texttt{ode45} with \texttt{atol=10\string^-6}, for time points $0$ to $100$ in increments of $0.5$, with $\beta=1$, $\mu_E=4$, $\mu_I=7$, $S(0)=0.9999$, $E_1(0)=0.0001$, $E_i(0)=0$ ($i>1$), $I_1(0)=0$, $I_i(0)=0$ ($i>1$), and $R(0)=0$. See Appendix \ref{sec:code} for details. } \label{fig:SEIRbench}
\end{figure}

To summarize these results, neither approach performs uniformly better than the other. In both comparisons, low dimensional models (i.e., smaller shape parameters and thus larger coefficients of variation) coded using the more  explicit (LCT-based) model formulation, like eqs. \eqref{eq:RMLCT} and \eqref{eq:SEIRLCT}, yielded numerical solutions \textit{faster} than mathematically equivalent models coded using the more general phase-type (GLCT-based) formulation, like eqs. \eqref{eq:RMPT} and \eqref{eq:SEIRPT}.

For higher dimensional models (i.e., larger shape parameters and thus smaller coefficients of variation), the phase-type (GLCT) formulation of these models outperformed their LCT-type counterparts. This is very likely the result of the matrix calculations used in the phase-type (GLCT) formulation of the models being more computationally efficient. 

It is noteworthy that the GLCT-based ODE function only needs to be coded once, as it is agnostic of the number of sub-state variables. In contrast, researchers typically hard-code the number of sub-states for such computations using an LCT-based model, and must write multiple ODE functions to consider model outputs using different shape parameters for the assumed Erlang distributions. 

In summary, the GLCT may allow for faster computing times for high dimensional systems and it can simplify writing code for ODE solvers since a single instance of the model can be used to simulate ODE models with an arbitrary number of dimensions.

\section{Discussion}\label{sec:discussion}

ODE models are widely used in the biological sciences and can often be viewed as mean field models for some (oftentimes, unspecified) stochastic state transition model. Such ODE models are sometimes criticized for their implicit assumption of exponentially distributed dwell times (e.g., exponentially distributed lifetimes of organisms), and their frequent lack of age or stage structure, which may not adequately capture important time lags in the system being modeled, such as the maturation times of organisms or latent periods in disease transmission models. 

In this paper, we have provided an overview of the Generalized Linear Chain Trick (GLCT), a relatively new tool modelers can use to improve upon existing ODE models to address these shortcomings, and we illustrate its utility using multiple examples. The GLCT extends the well-known Linear Chain Trick (LCT) to a broad family of probability distributions known as the phase-type distributions, and also clarifies in a straightforward way how mean field ODEs reflect underlying stochastic model assumptions when viewed from the perspective of continuous time Markov chains (CTMCs). Therefore, we have also provided an overview of CTMCs, and their absorption time distributions in particular. Importantly, the phase-type distributions comprise these absorption time distributions, and include a broad set of named probability distributions including exponential, Erlang, hypoexponential (generalized Erlang), hyperexponential, and Coxian distributions. Freely available statistical tools like BuTools \citep{BuTools2,BuToolswww} exist for fitting phase-type distributions to data, enabling modelers to build approximate empirical dwell time distributions into ODE models using the GLCT.

We have illustrated how to apply the GLCT by using it to derive extensions of two familiar models: the Rosenzweig-MacArthur Predator-Prey model, and the SEIR model of infectious disease transmission. We showed how two special cases of the generalized SEIR model can be constructed to accommodate additional complexity among infected individuals: the first case models a hospitalization scenario in which the transition to the hospitalized state has no impact on the distribution of the overall time individuals spend sick \citep[cf.][]{Feng2016}, and the second case models heterogeneity in the progression and severity of infection outcomes. These examples illustrate how the GLCT can be used to refine model assumptions in a rigorous manner, but without the need to explicitly derive mean field ODEs from stochastic models and/or mean field integral equations.

Lastly, we showed some of the potential computational benefits of using a GLCT-based approach to ODE model formulation by comparing the time it takes to compute numerical solutions of ODEs using the standard approach versus using a GLCT-based approach. We found that, for low dimensional models, the GLCT-based computations are slower than using a more traditional approach; however for higher dimensional models, the GLCT-based computations were faster. This improvement in computing time is likely the result of the computational efficiency of doing matrix and vector based computations. In addition to faster computation times, another benefit of the GLCT-based approach is that only one ODE model function needs to be written since it is agnostic of the model dimension.  In contrast, models that have been extended using the standard LCT typically have the number of sub-states (i.e., the shape parameter for the Erlang distributions) hard-coded, and therefore multiple model functions must be coded to explore different shape parameters. 

In conclusion, we hope this work encourages other researchers to think more carefully about underlying model assumptions when deriving ODE models. We also hope this work demonstrates the relative ease with which some basic intuition from Markov chain theory can be used to specify clear model assumptions from first principles, which can then be very quickly realized as one or more mean field ODE models using the GLCT \citep{Hurtado2019}.

\section*{Acknowledgements}
The authors thank Deena Schmidt, Jillian Kiefer, and the other members Mathematical Biology Lab Group at UNR for conversations and comments that improved this manuscript. %We also thank the editors of this issue for their generous accommodation to provide additional time to submit this work during the SARS-CoV-2/COVID-19 pandemic. 

\section*{Funding}
This work was supported by a grant awarded to PJH by the Sloan Scholars Mentoring Network of the Social Science Research Council with funds provided by the Alfred P. Sloan Foundation; and this material is based upon work supported by the National Science Foundation under Grant No. DEB-1929522. This work was partly motivated by PJH's participation in the ICMA-VII conference held at Arizona State University, October 12-14, 2019, with travel support provided to PJH from NSF grant \#DMS-1917512 awarded to the Organizing Committee of the ICMA-VII conference.

\section*{Disclosure statement}
The authors declare that they have no conflict of interest.

\begin{appendices} 
\numberwithin{equation}{section}
%\numberwithin{table}{section}
%\numberwithin{figure}{section}
\renewcommand{\theequation}{\Alph{section}\arabic{equation}}\setcounter{equation}{0} 

\section{R Code for Numerical Solutions to ODEs}\label{sec:code}

For the complete R code used to generate Figs. \ref{fig:RMbench} and \ref{fig:SEIRbench}, see the Electronic Supplements. The following R code shows a portion of that code, to illustrate how the GLCT-based model formulations differ from the LCT-based formulations.

\subsection{Rosenzweig-MacArthur Model \& Extensions}
		
\begin{lstlisting}
####################################################################################
## Benchmarking the Rosenzweig-MacArthur model numerical solutions
## ----------------------------------------------------------------
##  Case 1: Standard: No maturation delay, exponentially distributed predator lifetime
##  Case 2: Classic LCT (Erlang maturation time and adult life stage duration), hard-coded
##  Case 3: Phase-type/GLCT implementation of Case 2
####################################################################################

library(deSolve)
library(rbenchmark)

## Parameters

IC = c(N=1000,Y=10); 

parms = c(r = 1, K = 1000, a = 5, k = 500, chi = 0.5, mux = 0.5, muy = 1);

####################################################################################
## Function definitions

# Standard Rosenzweig-MacArthur model 
RM.ode <- function(tm,z,ps) {
r = ps[1]; K = ps[2]; a = ps[3]; k = ps[4]; chi = ps[5]; mux = ps[6]; muy = ps[7];

N=z[1] # prey population
Y=z[2] # predator population

dN = r*N*(1-N/K)-(a/(k+N))*Y*N
dY = chi*(a/(k+N))*N*Y-Y/muy
return(list(c(dN,dY)))
}


# Generalized Rosenzweig-MacArthur with phase-type maturation and adult-life-stage times
RMpt.ode <- function(tm,z,ps) {
r = ps[1]; K = ps[2]; a = ps[3]; k = ps[4]; chi = ps[5]; mux = ps[6]; muy = ps[6];
#kx = ps[['kx']] # Number of substates in E, and...
#ky = ps[['ky']] # ... I.

N = z[1] # prey population
X = z[1+(1:kx)] # immature predator population sub-states
Y = z[1+kx+(1:ky)] # mature predator population sub-states
P = onesyt %*% Y # total predators = sum Y_i

dN = r*N*(1-N/K)- a/(k+N)*P*N
dX = chi*(a/(k+N))*P[1,1]*N * ax + Axt %*% X
dY = ay %*% -onesxt %*% Axt %*% X + Ayt %*% Y

return(list(c(dN,as.numeric(dX),as.numeric(dY))))
}

# Initialize some variables for the phase-type R-M model above
RMpt.init <- function(ps) {
# Unpack some parameter values...
mux=ps[["mux"]] # mean maturation time
muy=ps[["muy"]] # mean time spent in mature life stage
kx <<- ps[['kx']] # Number of substates in E, and...
ky <<- ps[['ky']] # ... I.

# These yield Erlang distributions, where vectors
# ax = (1 0 ... 0) and matrices Ax are as follows...

ax = matrix(0, nrow=kx, ncol=1); ax[1] = 1;
Ax = kx/mux*(diag(rep(-1,kx),kx)); if(kx>1) for(i in 1:(kx-1)) {Ax[i,i+1] = kx/mux}

ay = matrix(0, nrow=kx, ncol=1); ay[1] = 1;
Ay = ky/muy*(diag(rep(-1,ky),ky)); if(ky>1) for(i in 1:(ky-1)) {Ay[i,i+1] = ky/muy}

# Initial conditions
z0=numeric(1+kx+ky) # initialize the state variable vector with 0s
z0[1] <- IC[["N"]]  # 1 in the initial exposed class 
z0[1+kx+1] <- IC[["Y"]]    # susceptibles = (PopSize - 1)/PopSize

# Set some global variables that the RMpt.ode function can access...
ax  <<- ax
Axt <<- t(Ax)
ay  <<- ay
Ayt <<- t(Ay)
onesxt <<- matrix(1,ncol=kx,nrow=1)
onesx <<- matrix(1,nrow=kx,ncol=1)
onesyt <<- matrix(1,ncol=ky,nrow=1)
onesy <<- matrix(1,nrow=ky,ncol=1)
ICs <<- z0
} # end of RMpt.init function

# Example R.-M. with hard-coded Erlang dwell times (standard LCT implementation)
RMlct1.ode <- function(tm,z,ps) {
r = ps[1]; K = ps[2]; a = ps[3]; k = ps[4]; chi = ps[5]; mux = ps[6]; muy = ps[7];
#kx = ps[['kx']] 
#ky = ps[['ky']] 

N = z[1]
P = onesyt %*% z[1+kx+(1:ky)] # total mature predators; sum Y_i

dN  = r*N*(1-N/K) - a/(k+N)*P*N
dX1 = chi*a/(k+N)*P*N - z[2]*kx/mux
dY1 = kx/mux*z[2] - z[3]*ky/muy

return(list(c(dN, dX1, dY1)))
}

RMlct2.ode <- function(tm,z,ps) {
r = ps[1]; K = ps[2]; a = ps[3]; k = ps[4]; chi = ps[5]; mux = ps[6]; muy = ps[7];
#kx = ps[['kx']] 
#ky = ps[['ky']] 

N = z[1]
P = onesyt %*% z[1+kx+(1:ky)] # total mature predators; sum Y_i

dN  = r*N*(1-N/K) - a/(k+N)*P*N
dX1 = chi*a/(k+N)*P*N - z[2]*kx/mux
dX2 = kx/mux*z[2] - z[3]*kx/mux
dY1 = kx/mux*z[3] - z[4]*ky/muy
dY2 = ky/muy*z[4] - z[5]*ky/muy

return(list(c(dN, dX1, dX2, dY1, dY2)))
}

# ... this pattern is repeated until 

RMlct20.ode <- function(tm,z,ps) {
r = ps[1]; K = ps[2]; a = ps[3]; k = ps[4]; chi = ps[5]; mux = ps[6]; muy = ps[7];
#kx = ps[['kx']] 
#ky = ps[['ky']] 

N = z[1]
P = onesyt %*% z[1+kx+(1:ky)] # total mature predators; sum Y_i

dN  = r*N*(1-N/K) - a/(k+N)*P*N
dX1 = chi*a/(k+N)*P*N - z[2]*kx/mux
dX2 = kx/mux*z[2] - z[3]*kx/mux
dX3 = kx/mux*z[3] - z[4]*kx/mux
dX4 = kx/mux*z[4] - z[5]*kx/mux
dX5 = kx/mux*z[5] - z[6]*kx/mux
dX6 = kx/mux*z[6] - z[7]*kx/mux
dX7 = kx/mux*z[7] - z[8]*kx/mux
dX8 = kx/mux*z[8] - z[9]*kx/mux
dX9 = kx/mux*z[9] - z[10]*kx/mux
dX10 = kx/mux*z[10] - z[11]*kx/mux
dX11 = kx/mux*z[11] - z[12]*kx/mux
dX12 = kx/mux*z[12] - z[13]*kx/mux
dX13 = kx/mux*z[13] - z[14]*kx/mux
dX14 = kx/mux*z[14] - z[15]*kx/mux
dX15 = kx/mux*z[15] - z[16]*kx/mux
dX16 = kx/mux*z[16] - z[17]*kx/mux
dX17 = kx/mux*z[17] - z[18]*kx/mux
dX18 = kx/mux*z[18] - z[19]*kx/mux
dX19 = kx/mux*z[19] - z[20]*kx/mux
dX20 = kx/mux*z[20] - z[21]*kx/mux
dY1 = kx/mux*z[21] - z[22]*ky/muy
dY2 = ky/muy*z[22] - z[23]*ky/muy
dY3 = ky/muy*z[23] - z[24]*ky/muy
dY4 = ky/muy*z[24] - z[25]*ky/muy
dY5 = ky/muy*z[25] - z[26]*ky/muy
dY6 = ky/muy*z[26] - z[27]*ky/muy
dY7 = ky/muy*z[27] - z[28]*ky/muy
dY8 = ky/muy*z[28] - z[29]*ky/muy
dY9 = ky/muy*z[29] - z[30]*ky/muy
dY10 = ky/muy*z[30] - z[31]*ky/muy
dY11 = ky/muy*z[31] - z[32]*ky/muy
dY12 = ky/muy*z[32] - z[33]*ky/muy
dY13 = ky/muy*z[33] - z[34]*ky/muy
dY14 = ky/muy*z[34] - z[35]*ky/muy
dY15 = ky/muy*z[35] - z[36]*ky/muy
dY16 = ky/muy*z[36] - z[37]*ky/muy
dY17 = ky/muy*z[37] - z[38]*ky/muy
dY18 = ky/muy*z[38] - z[39]*ky/muy
dY19 = ky/muy*z[39] - z[40]*ky/muy
dY20 = ky/muy*z[40] - z[41]*ky/muy

return(list(c(dN, dX1, dX2, dX3, dX4, dX5, dX6, dX7, dX8, dX9, dX10, dX11, dX12, dX13, dX14, dX15, dX16, dX17, dX18, dX19, dX20, dY1, dY2, dY3, dY4, dY5, dY6, dY7, dY8, dY9, dY10, dY11, dY12, dY13, dY14, dY15, dY16, dY17, dY18, dY19, dY20)))
}

# Set some of the ode() parameters for numerical solutions
mthd = "ode45"
atol= 1e-6
Tmax = 500
tms=seq(0,Tmax,length=500)


reps<-140

parms1 <- parms; parms1['kx']=1; parms1['ky']=1; RMpt.init(parms1); 
b1=benchmark(RM.1 ={ode(y=IC, times = tms, func = RM.ode, parms = parms, method = mthd, atol=atol)}, 
RMlct.1={ode(y=ICs, times=tms, func=RMlct1.ode, parms = parms1, method = mthd, atol=atol)},
RMpt.1 ={ode(y=ICs, times=tms, func=RMpt.ode,   parms = parms1, method = mthd, atol=atol)}, 
replications = reps)

parms2 <- parms; parms2['kx']=2; parms2['ky']=2; parms2; RMpt.init(parms2);
b2=benchmark(RM.2 ={ode(y=IC, times = tms, func = RM.ode, parms = parms, method = mthd, atol=atol)}, 
RMlct.2={ode(y=ICs, times=tms, func=RMlct2.ode, parms = parms2, method = mthd, atol=atol)},
RMpt.2 ={ode(y=ICs, times=tms, func=RMpt.ode,   parms = parms2, method = mthd, atol=atol)}, 
replications = reps)

# ... intermediate values omitted for breveity -- see electronic supplement for full code...

parms20 <- parms; parms20['kx']=20; parms20['ky']=20; parms20; RMpt.init(parms20);
b20=benchmark(RM.20 ={ode(y=IC, times = tms, func = RM.ode, parms = parms, method = mthd, atol=atol)}, 
RMlct.20={ode(y=ICs, times=tms, func=RMlct20.ode, parms = parms20, method = mthd, atol=atol)},
RMpt.20 ={ode(y=ICs, times=tms, func=RMpt.ode, parms = parms20, method = mthd, atol=atol)}, 
replications = reps) 

# Reorganize the data and change up some labeling...
# 1. Rearrange rows to the right order, save a new copy...

out <- rbind(b1,b2,b20)

# Display output
out


\end{lstlisting}	

\clearpage
\subsection{SEIR Model \& Extensions}

\begin{lstlisting}
####################################################################################
## Benchmarking SEIR model numerical solutions
## ----------------------------------------------
##  Case 1: Exponential dwell time distributions (dimension = 3 if R omitted)
##  Case 2: Classic LCT (Erlang dwell times), hard-coded number of substates
##  Case 3: SEIR w/ phase-type distributions, parameterized w/ Erlang distributions
####################################################################################
library(deSolve)
library(rbenchmark)

## Parameterization....
reps<-500
parms = c(b=1, muE=4, muI=7, cvE=1, cvI=1, kE=1, kI=1) 
parms2 = c(b=1, muE=4, cvE=1/sqrt(2+1e-10), muI=7, cvI=1/sqrt(2+1e-10), kE=2, kI=2)
parms2

parms3 = c(b=1, muE=4, cvE=1/sqrt(3.01), muI=7, cvI=1/sqrt(3.01)) 
parms3['kE'] <- floor((1/parms[['cvE']])^2)
parms3['kI'] <- floor((1/parms[['cvI']])^2)
parms3

####################################################################################
## Function definitions

# SEIR with Exponential dwell times
SEIR.ode <- function(tm,z,ps) {
	b=ps[["b"]]     # beta
	muE=ps[["muE"]] # mean time spent in E
	muI=ps[["muI"]] # mean time in I
	
	dS = -b*z[1]*z[3]
	dE =  b*z[1]*z[3] - z[2]/muE
	dI =  z[2]/muE - z[3]/muI
	dR = z[3]/muI
	
	return(list(c(dS, dE, dI, dR)))
}

# SEIR with hard-coded Erlang dwell times
SEIRlct1.ode <- function(tm,z,ps) {
	b=ps[["b"]]     # beta
	muE=ps[["muE"]] # mean time spent in E
	muI=ps[["muI"]] # mean time in I
	kE =1
	kI =1
	
	Itot = sum(z[1+kE+(1:kI)]) # sum(Ivec) 
	
	dS  = -b*z[1]*Itot
	dE1 =  b*z[1]*Itot  - z[2]*kE/muE
	dI1 =  z[2]*kE/muE  - z[3]*kI/muI
	dR  =  z[3]*kI/muI
	
	return(list(c(dS, dE1, dI1, dR)))
}

SEIRlct2.ode <- function(tm,z,ps) {
	b=ps[["b"]]     # beta
	muE=ps[["muE"]] # mean time spent in E
	muI=ps[["muI"]] # mean time in I
	kE =2
	kI =2
	
	Itot = sum(z[1+kE+(1:kI)]) # sum(Ivec) 
	
	dS  = -b*z[1]*Itot
	dE1 =  b*z[1]*Itot  - z[2]*kE/muE
	dE2 =  z[2]*kE/muE  - z[3]*kE/muE
	dI1 =  z[3]*kE/muE  - z[4]*kI/muI
	dI2 =  z[4]*kI/muI  - z[5]*kI/muI
	dR  =  z[5]*kI/muI
	
	return(list(c(dS, dE1, dE2, dI1, dI2, dR)))
}

## This pattern is continued until...

SEIRlct25.ode <- function(tm,z,ps) {
	b=ps[["b"]]     # beta
	muE=ps[["muE"]] # mean time spent in E
	muI=ps[["muI"]] # mean time in I
	kE = 25
	kI = 25
	
	Itot = sum(z[1+kE+(1:kI)]) # sum(Ivec) 
	
	dS  = -b*z[1]*Itot
	dE1 =  b*z[1]*Itot - z[2]*kE/muE
	dE2 =  z[2]*kE/muE  - z[3]*kE/muE
	dE3 =  z[3]*kE/muE  - z[4]*kE/muE
	dE4 =  z[4]*kE/muE  - z[5]*kE/muE
	dE5 =  z[5]*kE/muE  - z[6]*kE/muE
	dE6 =  z[6]*kE/muE  - z[7]*kE/muE
	dE7 =  z[7]*kE/muE  - z[8]*kE/muE
	dE8 =  z[8]*kE/muE  - z[9]*kE/muE
	dE9 =  z[9]*kE/muE  - z[10]*kE/muE
	dE10 =  z[10]*kE/muE  - z[11]*kE/muE
	dE11 =  z[11]*kE/muE  - z[12]*kE/muE
	dE12 =  z[12]*kE/muE  - z[13]*kE/muE
	dE13 =  z[13]*kE/muE  - z[14]*kE/muE
	dE14 =  z[14]*kE/muE  - z[15]*kE/muE
	dE15 =  z[15]*kE/muE  - z[16]*kE/muE
	dE16 =  z[16]*kE/muE  - z[17]*kE/muE
	dE17 =  z[17]*kE/muE  - z[18]*kE/muE
	dE18 =  z[18]*kE/muE  - z[19]*kE/muE
	dE19 =  z[19]*kE/muE  - z[20]*kE/muE
	dE20 =  z[20]*kE/muE  - z[21]*kE/muE
	dE21 =  z[21]*kE/muE  - z[22]*kE/muE
	dE22 =  z[22]*kE/muE  - z[23]*kE/muE
	dE23 =  z[23]*kE/muE  - z[24]*kE/muE
	dE24 =  z[24]*kE/muE  - z[25]*kE/muE
	dE25 =  z[25]*kE/muE  - z[26]*kE/muE
	dI1 =  z[26]*kE/muE  - z[27]*kI/muI
	dI2 =  z[27]*kI/muI  - z[28]*kI/muI
	dI3 =  z[28]*kI/muI  - z[29]*kI/muI
	dI4 =  z[29]*kI/muI  - z[30]*kI/muI
	dI5 =  z[30]*kI/muI  - z[31]*kI/muI
	dI6 =  z[31]*kI/muI  - z[32]*kI/muI
	dI7 =  z[32]*kI/muI  - z[33]*kI/muI
	dI8 =  z[33]*kI/muI  - z[34]*kI/muI
	dI9 =  z[34]*kI/muI  - z[35]*kI/muI
	dI10 =  z[35]*kI/muI  - z[36]*kI/muI
	dI11 =  z[36]*kI/muI  - z[37]*kI/muI
	dI12 =  z[37]*kI/muI  - z[38]*kI/muI
	dI13 =  z[38]*kI/muI  - z[39]*kI/muI
	dI14 =  z[39]*kI/muI  - z[40]*kI/muI
	dI15 =  z[40]*kI/muI  - z[41]*kI/muI
	dI16 =  z[41]*kI/muI  - z[42]*kI/muI
	dI17 =  z[42]*kI/muI  - z[43]*kI/muI
	dI18 =  z[43]*kI/muI  - z[44]*kI/muI
	dI19 =  z[44]*kI/muI  - z[45]*kI/muI
	dI20 =  z[45]*kI/muI  - z[46]*kI/muI
	dI21 =  z[46]*kI/muI  - z[47]*kI/muI
	dI22 =  z[47]*kI/muI  - z[48]*kI/muI
	dI23 =  z[48]*kI/muI  - z[49]*kI/muI
	dI24 =  z[49]*kI/muI  - z[50]*kI/muI
	dI25 =  z[50]*kI/muI  - z[51]*kI/muI
	dR  =  z[51]*kI/muI
	
	return(list(c(dS, dE1, dE2, dE3, dE4, dE5, dE6, dE7, dE8, dE9, dE10, dE11, dE12, dE13, dE14, dE15, dE16, dE17, dE18, dE19, dE20, dE21, dE22, dE23, dE24, dE25, dI1, dI2, dI3, dI4, dI5, dI6, dI7, dI8, dI9, dI10, dI11, dI12, dI13, dI14, dI15, dI16, dI17, dI18, dI19, dI20, dI21, dI22, dI23, dI24, dI25, dR)))
}

# SEIR with Erlang dwell times in a GLCT framework
SEIRpt.ode <- function(tm,z,ps) {
	b=ps[["b"]]     # beta
	muE=ps[["muE"]] # mean time spent in E
	muI=ps[["muI"]] # mean time in I
	kE =ps[['kE']]
	kI =ps[['kI']]
	Itot = sum(z[1+kE+(1:kI)]) # sum(Ivec) 
	dS = -b*z[1]*Itot
	dE =  b*z[1]*Itot                        *  aE + AEt %*% z[1+(1:kE)]  # t(AE) %*% Evec
	dI =  aI %*% (-OnesE%*%AEt %*% z[1+(1:kE)]) + AIt %*% z[1+kE+(1:kI)] #  t(AI) %*% Ivec
	dR = as.numeric(-OnesI%*%AIt %*% z[1+kE+(1:kI)]) 
	
	return(list(c(dS, as.numeric(dE), as.numeric(dI), dR)))
} 

# Function to initialize some objects used by SEIRpt.ode()
SEIRpt.init <- function(ps) {
	# Unpack some parameter values...
	muE=ps[["muE"]] # mean time spent in E
	muI=ps[["muI"]] # mean time in I
	kE = ps[['kE']] # Number of substates in E, and...
	kI = ps[['kI']] # ... I.
	
	# These are Erlang distributions framed in a Phase-type distribution context,
	# where vector a = (1 0 ... 0) and matrix A is as follows...
	
	aE = matrix(0,nrow = kE, ncol=1); aE[1] = 1;
	AE = kE/muE*(diag(rep(-1,kE),kE)); if(kE>1) for(i in 1:(kE-1)) {AE[i,i+1] = kE/muE}
	
	aI = matrix(0,nrow = kI, ncol=1); aI[1] = 1;
	AI = kI/muI*(diag(rep(-1,kI),kI)); if(kI>1) for(i in 1:(kI-1)) {AI[i,i+1] = kI/muI}
	
	# Initial conditions
	PopSize=10000
	z0=numeric(kE+kI+2) # initialize the 1+kE+kI+1 state variables
	z0[2] <- 1/PopSize  # 1 in the initial exposed class 
	z0[1] <- 1-z0[2]    # susceptibles = (PopSize - 1)/PopSize
	
	# Set some global variables...
	aE  <<- aE
	AEt <<- t(AE)
	aI  <<- aI
	AIt <<- t(AI)
	OnesE <<- matrix(1,ncol=kE,nrow=1)
	OnesI <<- matrix(1,ncol=kI,nrow=1)
	ICs <<- z0
}

############################################################################
## Example plot
library(ggplot2)

Tmax = 100
tms=seq(0,Tmax,length=200)
SEIRpt.init(parms) # set's some initial conditions in ICs, and some matrices needed for SEIRpt.ode()

out = ode(y=ICs, times = tms, func = SEIRpt.ode, parms = parms, method = "ode23");
c(Pos=sum(out[,-1]>=0), Neg=sum(out[,-1]<0) )

matplot(tms, out[,-1],type="l",lty=1)

out = ode(y=ICs, times = tms, func = SEIRpt.ode, parms = parms, method = "ode45");
c(Pos=sum(out[,-1]>=0), Neg=sum(out[,-1]<0) )

matplot(tms, out[,-1],type="l",lty=2,add=TRUE)

out = ode(y=ICs, times = tms, func = SEIRpt.ode, parms = parms, method = "lsodes");
c(Pos=sum(out[,-1]>=0), Neg=sum(out[,-1]<0) )

matplot(tms, out[,-1],type="l",lty=2,add=TRUE)

out = ode(y=ICs, times = tms, func = SEIRpt.ode, parms = parms, method = "vode");
c(Pos=sum(out[,-1]>=0), Neg=sum(out[,-1]<0) )

matplot(tms, out[,-1],type="l",lty=2,add=TRUE)


############################################################################
## Benchmark...

SEIR.ode <- compiler::cmpfun(SEIR.ode)
SEIRpt.ode <- compiler::cmpfun(SEIRpt.ode)
SEIRlct1.ode <- compiler::cmpfun(SEIRlct1.ode)
SEIRlct2.ode <- compiler::cmpfun(SEIRlct2.ode)
SEIRlct25.ode <- compiler::cmpfun(SEIRlct25.ode)

mthd = "ode45";  atol= 1e-6
IC=c(S=0.9999, E=0.0001, I=0, R=0) # for SEIR.ode()

parms1 <- parms; parms1['kE']=1; parms1['kI']=1; parms1; SEIRpt.init(parms1)
b1=benchmark(SEIR.1 ={ode(y=IC, times = tms, func = SEIR.ode, parms = parms, method = mthd, atol=atol)}, 
SEIRlct.1={ode(y=ICs, times=tms, func=SEIRlct1.ode, parms = parms1, method = mthd, atol=atol)},
SEIRpt.1 ={ode(y=ICs, times=tms, func=SEIRpt.ode,   parms = parms1, method = mthd, atol=atol)}, 
replications = reps)

parms2 <- parms; parms2['kE']=2; parms2['kI']=2; parms2; SEIRpt.init(parms2)
b2=benchmark(SEIR.2 ={ode(y=IC, times = tms, func = SEIR.ode, parms = parms, method = mthd, atol=atol)}, 
SEIRlct.2={ode(y=ICs, times=tms, func=SEIRlct2.ode, parms = parms2, method = mthd, atol=atol)},
SEIRpt.2 ={ode(y=ICs, times=tms, func=SEIRpt.ode,   parms = parms2, method = mthd, atol=atol)}, 
replications = reps)

parms25 <- parms; parms25['kE']=25; parms25['kI']=25; parms25; SEIRpt.init(parms25)
b25=benchmark(SEIR.25 ={ode(y=IC, times = tms, func = SEIR.ode, parms = parms, method = mthd, atol=atol)}, 
SEIRlct.25={ode(y=ICs, times=tms, func=SEIRlct25.ode, parms = parms25, method = mthd, atol=atol)},
SEIRpt.25 ={ode(y=ICs, times=tms, func=SEIRpt.ode,   parms = parms25, method = mthd, atol=atol)}, 
replications = reps)

b1; b2; b25

# First, reorganize the data and change some labeling...
# 1. Rearrange rows to the right order, save a new copy...
out <- rbind(b1,b2,b25)

out$N <- stringr::str_split_fixed(out$test,'\\.',2)[,2]
out$Model <- stringr::str_split_fixed(out$test,'\\.',2)[,1]
# out$Model <- gsub("SEIR$","SEIR (Exponential)",out$Model)
# out$Model <- gsub("lct"," (Erlang / LCT)",out$Model)
# out$Model <- gsub("pt"," (Erlang as Phase-Type / GLCT)",out$Model)
out$Model <- gsub("SEIR$","Standard: Exponential latent & infetious periods",out$Model)
out$Model <- gsub("lct","LCT:    Erlang latent & infectious periods",out$Model)
out$Model <- gsub("pt","GLCT: Erlang latent & infectious periods (phase-type formulation)",out$Model)
out$Model <- gsub("SEIR",'',out$Model)
out$Model <- factor(out$Model, levels = unique(out$Model) )
out$N <- as.numeric(out$N)
out
\end{lstlisting}

\end{appendices}

%\clearpage
%\bibliography{./myrefs}
\printbibliography

\end{document}